%% file: writeup.tex
\documentclass[runningheads]{elsarticle}

\usepackage{color}
\usepackage{times, amssymb, amsmath, comment}
\usepackage{hyperref}
\usepackage[boxed,linesnumbered]{algorithm2e}
\usepackage{verbatim}

\newtheorem{theorem}{Theorem}
\newtheorem{lemma}[theorem]{Lemma}

\newtheorem{corollary}[theorem]{Corollary}

\newtheorem{defi}[theorem]{Definition}
\newtheorem{obs}[theorem]{Observation}
\newproof{proof}{Proof}
\newproof{proofofclm}{Proof of Claim}
\newtheorem{note}[theorem]{Note}
\newtheorem{fact}[theorem]{Fact}

\begin{document}

\begin{frontmatter}

\pagestyle{headings}
\title{Testing whether the Uniform Distribution is a Stationary Distribution}
\author[label1]{ Sourav Chakraborty}
\author[label1]{ Akshay Kamath}
\author[label1]{ Rameshwar Pratap\footnote{{Corresponding author:}
\newline \emph{This work done when the author was pursuing his PhD from Chennai Mathematical Institute.}
\newline Present/corresponding address: 122/7 PAC Colony, Naini, Allahabad, UP, India.
\newline Contact Number: +91 9953842289
}}



\address[label1]{
Chennai Mathematical Institute,\\
Chennai, India.\\
e-mail:\{sourav,adkamath,rameshwar\}@cmi.ac.in
}


\begin{abstract}

A random walk on a directed graph generates a Markov chain on the vertices of the graph. 
An important question that  often arises in the context of  Markov chains is, whether 
the uniform distribution on the vertices of the graph is a stationary distribution. 
A stationary distribution of a Markov chain is a global property of the graph. This 
leads to the belief that whether a particular distribution is a stationary 
distribution of a Markov chain depends on the global property of that Markov chain. 
In this paper for a directed graph whose underlying
undirected graph is regular, we prove that whether the uniform distribution on the vertices of the
graph is a stationary distribution, depends on a local property of the graph, namely if $(u,v)$ is a
directed edge, then out-degree(u) is equal to in-degree(v).


\hspace{15pt} This result also has an application to the problem of testing whether 
a given distribution is uniform or ``far'' from being uniform. 
If the distribution is the stationary distribution of the
lazy random walk on a directed graph and the graph is given as an input, then
how many bits (orientations) of the input graph does one need to query in order to decide whether the distribution
is uniform or ``far''\footnote{Here, farness does not imply any statistical distance between the distributions. 
Rather, it specifies  the distance between the orientations - i.e. the minimum number of edges that need to 
be reoriented such that the stationary distribution obtained by a random walk on the resulting graph 
(obtained after reorientation of edges) is uniform.} from it? This is a problem of graph property testing, and
we consider this problem in the orientation model. 
We reduce this problem to testing Eulerianity in the orientation model.

\end{abstract}

\begin{keyword}
 Markov Chain; Property Testing; Orientation Model; Stationary Distribution.
\end{keyword}

\end{frontmatter}

\input{introduction}

\input{preliminaries}

\input{structure}
\input{testing}

\input{conclusion}
\bibliographystyle{plain}
\bibliography{reference}

\end{document}

%% file: introduction.tex

\section{Introduction}

{\em Spectral properties} of undirected graphs have been well studied and well understood 
\cite{chungspectral}. However, there has been less success   
in the study of the same in the case of directed graphs, possibly due the non-symmetric structure 
associated with its adjaceny matrix. In this work, we attempt to  understand the spectral properties 
of   Markov chains obtained by a random walk on a directed graph.
Markov chains are one of the most important   structures in Theoretical 
Computer Science. 
The most significant characteristics of a Markov chain are  its stationary distribution and  mixing time. 
It is an interesting problem to test if a particular distribution is a stationary 
distribution of a given Markov chain.  Since a stationary distribution of a Markov chain is a global property of the graph,  
this leads to the belief that 
whether a particular distribution is a stationary distribution of a Markov chain 
depends on its  global structure. 

 In this paper, we focus on the Markov chain 
obtained by a random walk on a directed graph. 
 We prove contrary to aforementioned belief that if the graph 
 is regular, then whether the uniform distribution is a stationary distribution depends on a local property of the graph.
  The following theorem, which is the main result of this paper, 
  is a statement about 
  that local property. 

\begin{theorem}\label{theorem:main2}
Let $\overrightarrow{G} = (V,\overrightarrow{E})$ be a directed graph such that  the total degree 
(\textit{i.e.}, \textrm{Indegree(v)+ Outdegree(v)}) for every vertex $v \in V$ is the same. Then 
the uniform distribution on the vertices of $\overrightarrow{G}$ is a stationary distribution 
of the Markov chain (generated by a random walk on $\overrightarrow{G}$) if and only if 
the graph has the following properties:
\begin{enumerate}
\item for all $v \in V$, \textrm{Indegree(v)} $\neq 0$ and \textrm{Outdegree(v)} $\neq 0$,
\item for all $(u,v) \in \overrightarrow{E}$, \textrm{Outdegree(u)=Indegree(v)}.
\end{enumerate}
\end{theorem}

As an application of this result, we design a testing algorithm to test whether the uniform distribution  is the stationary distribution of the Markov chain generated by a 
\textit{lazy random walk} \footnote{ A lazy random walk (starting from a particular vertex) on 
a directed graph is a random walk in which at each time step, the walk stays where it is with 
probability $\frac{1}{2}$ or moves according to the usual random walk. Moreover, 
it converges to a unique stationary distribution.} on $\overrightarrow{G}$.

\subsection{Application to Property Testing of Distributions}

In property testing, the goal is to look at a very small fraction of the input and decide 
whether the input has a certain property, or it is ``far'' from satisfying the property. 
For a given distance parameter $0<\epsilon<1$, we say that the input is an $\epsilon$-far from satisfying
the property if one has to change at least $\epsilon$ fraction of the input to make the input
satisfy the property. 
Theorem~\ref{theorem:main2} also has an application to the problem of testing whether a given distribution
is uniform or ``far'' from being uniform. More precisely, if the distribution is the stationary distribution of the
 lazy random walk 
  on a directed graph and the graph is given as 
an input, then how many bits of the input graph do one need to query in order to decide whether the distribution
is uniform or ``far'' from it? We consider this problem in the \textbf{orientation model} (see \cite{halevy}).
In this  model, the underlying undirected graph $G = (V,E)$ is known in advance, 
and the orientation of the edges has to be queried. The graph is said to be ``$\epsilon$-far" from satisfying the property $P$ if 
one has to reorient at least an $\epsilon$ fraction of the edges to make the graph satisfy the property.
We reduced this problem to testing Eulerianity in the orientation model. We use the results of \cite{euler} to obtain an algorithm that incurs sublinear cost for the 
above problem. We present this part of our result in  Section~\ref{section:testing}.

%% file: preliminaries.tex
\section{Preliminaries}
\subsection{Graph Notations}
Throughout the paper, we will be dealing with directed graphs (possibly with multiple edges between any two vertices) 
in which each edge is directed only in one direction.  To avoid confusion, we will call them \textbf{oriented graphs} 
because each edge is oriented 
and is not bidirectional. We will denote the oriented graph by $\overrightarrow{G} = (V, \overrightarrow{E})$, 
and the underlying undirected graph (that is, when the directions on the edges are removed) by 
$G = (V,E)$. For a vertex $v \in V$, the in-degree and the out-degree of $v$ in $\overrightarrow{G}$ 
are denoted by $d^-(v)$ and $d^+(v)$ respectively. An oriented graph $\overrightarrow{G} = (V,\overrightarrow{E})$ is called 
a \emph{degree-{$\Delta$}} oriented graph if for all $v\in V$,
$d^-(v)\mbox{ + } d^+(v) = \Delta$. In this paper, we will be focusing on \emph{degree}-$\Delta$ oriented graphs. 

\subsection{Markov Chains Preliminaries}
\begin{fact}
 A Markov chain is a stochastic process on a set of states given by a transition matrix. Let $S$ be the set of states 
with $|S| = n$. Then, the transition matrix $T$ is an $n\times n$ matrix with entries from the positive reals; 
the rows and columns are indexed by the states; the $(u,v)$-th entry $T_{u,v}$ of the matrix denotes the probability of 
transition from state $u$ to state $v$. Since $T$ is stochastic, $\sum_v T_{u,v}$ must be $1$. 
A distribution $\mu:S \to \mathbb{R}^+$ on the vertices is said to be stationary if for all 
vertices $v$,$$\sum_{v} \mu(u) T_{u,v} = \mu(v).$$
\end{fact}

\begin{fact}

 If $\overrightarrow{G}$ is an oriented graph then a random walk on $\overrightarrow{G}$ defines a Markov 
chain, where, the states are the vertices of the graph; the probability to traverse an edge $(u,v)$ is given 
by the quantity $p_{u,v} = \frac{1}{d^+(u)}$; and hence, the transition 
probability $T_{u,v}$ from vertex $u$ to vertex $v$ is $p_{u,v}$ times the number of edges between $u$ and $v$. 
The uniform distribution on the vertices of $\overrightarrow{G}$ is a stationary distribution for this Markov chain if and only if for all $v\in V$,
$$ \sum_{u:(u,v) \in \overrightarrow{E}} p_{u,v} = 1 = \sum_{w:(v,w) \in \overrightarrow{E}} p_{v,w}.$$
 
\end{fact}

\begin{note}
 In this paper, we will only  consider  Markov chains that arise from random walks 
on $\overrightarrow{G}$, where $\overrightarrow{G}$ is an oriented graph. 
\end{note}

%% file: structure.tex
\section{Structure of Graphs with Uniform Stationary Distribution}\label{section:structure}


The following theorem is a rephrasing of   Theorem~\ref{theorem:main2}.

\begin{theorem}\label{theorem:uniformDist}
Let $\overrightarrow{G} = (V,\overrightarrow{E})$ be a degree-$\Delta$ oriented graph,
then the uniform distribution on the vertices of  $\overrightarrow{G}$ is a stationary 
distribution of the  Markov chain 
if and only if for all $v\in V$, both $d^-(v), d^+(v) \neq 0$
and for all $(u,v) \in  \overrightarrow{E}$,
\[d^{+}(u) = d^-(v).\]
\end{theorem}

\begin{proof}
First of all, recall that the uniform distribution is a stationary
distribution for $\overrightarrow{G}$, if and only if for all $v \in V$,
$$\sum_{u :(u,v) \in \overrightarrow{E}}p_{u,v} = 1 = \sum_{w :(v,w) \in \overrightarrow{E}}p_{v,w},$$
where $p_{u,v}$ is the transition probability from vertex $u$ to vertex $v$, and $p_{u,v} = \frac{1}{d^+(u)}$.

Thus, if the graph $\overrightarrow{G}$ has the property that for all $(u,v) \in  \overrightarrow{E}$,
$d^+(u) = d^-(v)$, then
$$\sum_{u :(u,v) \in \overrightarrow{E}}p_{u,v} = \sum_{u :(u,v) \in \overrightarrow{E}}\frac{1}{d^+(u)} =
\sum_{u :(u,v) \in \overrightarrow{E}}\frac{1}{d^-(v)} = 1.$$ The last equality holds because the summation is over all the
edges entering $v$ (which is non-empty), and thus there are $d^-(v)$ number of items in the summation.
Similarly, $$\sum_{w :(v,w) \in \overrightarrow{E}}p_{v,w} = \sum_{w :(v,w) \in \overrightarrow{E}}\frac{1}{d^+(v)} = 1.$$
Therefore, if the graph $\overrightarrow{G}$ has the property that for all $(u,v) \in  \overrightarrow{E}$, $d^+(u) = d^-(v)$, then 
the uniform distribution is a stationary distribution of the Markov chain.

Now, let us prove the other direction, that is, let us assume that the
uniform distribution is a stationary distribution of the Markov
chain. We prove this direction by  contradiction.  The central idea of the 
proof is the notion of {\em``degree-alternating''} path (see Definition~\ref{defi:degalternating} ). 
Now, if there is an edge $(u, v)$ such that $d^+(u)\neq d^-(v)$, then we show 
that there is a {\em degree-alternating} path with infinitely many {\em unbalanced edges} 
(see Definition~\ref{defi:unbalanced}). Further, by Lemma~\ref{lemma:decSequence}, if there is a {\em degree-alternating} 
path with infinitely many {\em unbalanced edges}, then there is a vertex having negative in-degree or out-degree, 
which is not possible, and we get a contradiction.

We start the proof with a 
couple of easy observations. 


\begin{obs}\label{obs:two}
If the uniform distribution is a stationary distribution of
the Markov chain then:
\begin{enumerate}
  \item Every connected component in the underlying graph is strongly connected.
\item For a vertex $v$,  there exists an edge $(u, v)$
  such that $d^+(u) < d^-(v)$ if and only if there exists an edge $(w,
  v)$ such that $d^+(w) > d^-(v)$. 
\end{enumerate}
\end{obs}


Using the above observations, we prove our theorem by
contradiction. 

\begin{defi}\label{defi:degalternating}
A path $\{v_i\}_{i=0}^t$  is called a ``degree-alternating''
path if  the following  conditions are satisfied:
\begin{itemize}
\item For all $i\geq 0$,  $(v_{i+1}, v_{i})\in \overrightarrow{E}$, 
\item For all $i\geq 0$, $d^+(v_{2i+1})= \min\left\{d^+(w)\ :\ {(w,v_{2i})\in \overrightarrow{E}}\right\}$,
\item For all $i> 0$, $d^+(v_{2i})= \max\left\{d^+(w)\ :\ {(w,v_{2i -
        1})\in \overrightarrow{E}}\right\}$.
\end{itemize}
Note that the above  degree-alternating path can have repeated vertices (or edges). 
\end{defi}

\begin{defi}\label{defi:unbalanced}
We call an edge $(u,v)$ ``unbalanced'' if $d^+(u) \neq d^-(v)$.   Also, 
we call a vertex $w$ ``unbalanced'' if there is an edge $(w', w)$ such 
that $d^+(w') \neq d^-(w)$.  
\end{defi}

Now, we will show that if a strongly connected graph has one
``unbalanced'' edge, then there must be a ``degree-alternating''  path
with infinitely many unbalanced edges.  And this along with the
Lemma~\ref{lemma:decSequence}  will give a contradiction.

If the graph does not have any degree-alternating path with infinitely
many unbalanced edges, then there must be a degree-alternating path
ending at an unbalanced edge $(u,v)$ such that the path cannot be
extended to a longer path with more unbalanced edges. We show that
this is not possible by showing that we can extend a given degree-alternating path to a 
longer degree-alternating path, with at least one more unbalanced
edge. 

Note that if $u$ is an unbalanced vertex, then by Observation~\ref{obs:two}, we can extend
the degree-alternating path ending at $(u, v)$ by one more unbalanced
edge $(v, w)$ such that $ d^-(v) \neq d^+(w)$. Similarly, if there is a path
from an unbalanced vertex to $u$, then we can extend the
degree-alternating path with an edge $(u, w')$ such that $ d^-(u) \neq d^+(w')$. 
Further by Observation~\ref{obs:two}, since  underlying graph is 
strongly connected, there is a path from $v$ to $u$, then using above arguments
we can surely extend the degree-alternating path. 
As a consequence, we get  a
degree-alternating path with infinitely many unbalanced edges.

Now if we apply to Lemma~\ref{lemma:decSequence} on the
degree-alternating path with infinitely many unbalanced edges, and from
the fact that the in-degree or out-degree of a  vertex cannot be
negative, we get a contradiction. This completes the proof of the theorem.
\end{proof}

\begin{lemma}\label{lemma:decSequence} 
Let $\{v_i\}_{i=0}^t$ be a ``degree-alternating" path.
Suppose we define we define a new sequence $\{s_i\}_{i=0}^t$ of positive integers as following:
\begin{itemize}
\item For all $k\geq 0$, $s_{2k}$=$d^-(v_{2k})$, and for all $k \geq 0$, $s_{2k+1}$=$d^+(v_{2k+1})$.
\end{itemize}
Then this sequence of positive integers is a non-increasing sequence.
Moreover, if $v_{i}$ and $v_{i+1}$ are two consecutive vertices in the sequence such that
$d^-(v_{i+1})\neq d^+(v_i)$, then $s_{i+1} < s_i$.
\end{lemma}





\begin{proof}
Since we have assumed that the uniform distribution is a stationary distribution of 
the Markov chain, then for all vertices $v$,
$$\sum_{u:(u,v) \in \overrightarrow{E}}\frac{1}{d^+(u)} = 1 = \sum_{w:(v,w) \in \overrightarrow{E}}\frac{1}{d^+(v)}.$$

Let us first prove that in the sequence $\{s_i\}_{i=0}^t$, $s_{2i} \geq s_{2i +1}$.
\\
Since $d^+(v_{2i+1})= \min\left\{d^+(w)\ :\ {(w,v_{2i})\in \overrightarrow{E}}\right\}$,
 \begin{equation}\label{eqn:1st}
1 = \sum_{(u,v_{2i}) \in \overrightarrow{E}}\frac{1}{d^+(u)} \leq \frac{d^-(v_{2i})}{d^+(v_{2i+1})},
\end{equation}
and hence, we have $d^-(v_{2i}) \geq d^+(v_{2i+1})$ which by definition gives $s_{2i} \geq s_{2i+1}$.

Now, let us also prove that in the sequence $\{s_i\}_{i=0}^t$, $s_{2i-1} \geq s_{2i}$. By definition, this is the same as proving
$d^+(v_{2i-1}) \geq d^-(v_{2i})$. Since we have assumed that the graph $\overrightarrow{G}$ is a \emph{degree}-$\Delta$ graph, proving that 
$s_{2i-1} \geq s_{2i}$ is same as proving $d^-(v_{2i-1}) \leq d^+(v_{2i})$.

Similar to previous case, since $$d^+(v_{2i})= \max\left\{d^+(w)\ :\ {(w,v_{2i-1})\in \overrightarrow{E}}\right\}$$ and 
\begin{equation}
\label{eqn:2nd} 1 = \sum_{(u,v_{2i-1}) \in \overrightarrow{E}}\frac{1}{d^+(u)} \geq \frac{d^-(v_{2i-1})}{d^+(v_{2i})}.
\end{equation}  Thus, we have
$s_{2i-1} \geq s_{2i}$.

Note that the inequalities in  Equations~\ref{eqn:1st} and~\ref{eqn:2nd} are  strict inequalities if $d^+(v_{2i+1}) \neq d^-(v_{2i})$ 
and $d^+(v_{2i-1}) \neq d^-(v_{2i})$ respectively. Thus, if for any $i$, $d^+(v_{i+1}) \neq d^-(v_{i})$, then $s_{i+1} < s_i$.
This completes the proof of the lemma.
\end{proof}

From  Theorem~\ref{theorem:uniformDist}, we can also obtain the
following corollary. 
 
\begin{corollary}\label{lemma:bipartite}
Let $\overrightarrow{G} = (V, \overrightarrow{E})$ be a connected degree-$\Delta$ oriented 
graph. Then, the uniform distribution on the vertices of  $\overrightarrow{G}$ is a stationary 
distribution of the Markov chain, if and only if the following conditions apply:
\begin{enumerate}
\item  If the underlying undirected graph $G=(V,E)$ is non-bipartite, then the graph $\overrightarrow{G}$ is Eulerian.
\item  If $G$ is bipartite with bipartition $V_1 \cup V_2 = V$, then $|V_1| = |V_2|$ and the in-degree 
       of all vertices in one part will be same, and it will be equal to out-degree of all vertices in the other part.
\end{enumerate}
\end{corollary}

\begin{proof}
From  Theorem~\ref{theorem:uniformDist}, it follows that  the uniform distribution on the vertices of $\overrightarrow{G}$ 
is a stationary distribution of the  Markov chain  if and only if for all $(u,v) \in \overrightarrow{E}$, we have
\begin{equation}\label{eqn:struc} d^+(u) = d^-(v).\end{equation}

Suppose  $G$ has 
   a path of length $2$ between vertices $u$ and $w$ through $v$ (that is
$(u, v)$ and $(v, w)$ are  edges in  $G$). 
Then it is easy to verify that in all four possible cases (depending on the
orientation of the edges $(u, v)$ and $(v, w)$),  we have $d^+(u) =
d^+(w)$ (by Equation~\ref{eqn:struc}).  
Similarly,  it is also easy to see that
  if there a path of even length between $u$ and $w$ in
  $G$ (a path from $u$ to $w$ using even number of edges), then we have $d^+(u) = d^+(w)$. 

From this we can prove both parts of the corollary. 


Let $v$ be a vertex of any odd cycle in $G$,  and
 $(u, v)$ be an edge belong to that cycle, then $d^+(v) = d^+(u) = d^-(v)$
(the last equality follows from Equation~\ref{eqn:struc}).  
And now note that, if there exist a vertex having in-degree  equal to
out-degree and Equation~\ref{eqn:struc} holds, then  all the  
vertices in the strongly-connected component have in-degree  equal to
out-degree. Further, since $G$ is connected, 
$\overrightarrow{G}$ must be Eulerian. Also note that if $\overrightarrow{G}$ is Eulerian,
then the uniform distribution is a stationary distribution. 

For the case of bipartite graphs, if $V_1$ and $V_2$ is the
bi-partition, then any two vertices in $V_1$ are connected by an even
length path, and hence all the vertices in $V_1$ has the same
out-degree, and since all the vertices have the same total degree so all
the vertices in $V_1$ must have same in-degree also. Similarly, all the
vertices in $V_2$ have the same out-degree and the same in-degree. 
And if there is a directed edge from a vertex in $V_1$ to vertex in
$V_2$ (or vice versa), then  from Equation~\ref{eqn:struc}, we have the
out-degree of vertices in $V_1$ equals the in-degree of the vertices in
$V_2$.
\end{proof}

 Theorem~\ref{theorem:uniformDist} and  Corollary~\ref{lemma:bipartite} have an application to property 
testing. We present this application in the next section. 

%% file: testing.tex


\section{Application to Property Testing}\label{section:testing}
A property of a graph that is invariant under graph isomorphism is called a graph property. 
Testing of graph properties has been a very active topic in the area of property testing 
(see \cite{Fischer04}, \cite{ron}), and  have been studied under various
query models, for example: dense-graph models, sparse-graph models, orientation model etc.
Note that, whether the uniform distribution on the vertices of $\overrightarrow{G}$ is a stationary distribution of the
Markov chain, is a graph property. 
Thus, the problem of distinguishing between whether the stationary distribution of the Markov chain 
is the uniform distribution, or is ``far'' from it, is a question of testing graph properties. 
Here, this question has been framed in the \textit{orientation model} (defined in the next 
subsection).  Some interesting graph properties like connectivity \cite{stcon} and Eulerianity
 \cite{euler} have been studied in this model. Using  Theorem~\ref{theorem:main2}, we show that for 
 both bipartite and non-bipartite graphs, testing (in the orientation model) whether the uniform 
 distribution on vertices of $\overrightarrow{G}$ is the stationary distribution of the Markov 
 chain (generated by a \textit{lazy random walk} on $\overrightarrow{G}$), can be reduced to 
 testing if the graph is Eulerian. Using  algorithms from \cite{euler} of  testing Eulerianity 
 in the orientation model, we obtain various bounds on the
query complexity for testing uniformity of the stationary distribution.

In \cite{euler} it is shown that if $G$ is an $\alpha$-expander, then it is possible to test, 
in the orientation model, whether $G$ is Eulerian by performing $O(\Delta/\alpha)$ queries.
From our result, it implies that if $G$ is an $\alpha$-expander, then testing uniformity of the 
stationary distribution can be done with $O(\Delta/\alpha)$ queries. 
Since $1/\alpha$ is also a measure of the mixing time of the random walk, it implies that 
the query complexity for testing uniformity of the stationary distribution 
is directly proportional to the mixing time of the Markov chain.

\subsection{Property Testing in the Orientation Model}
Given an oriented graph $\overrightarrow{G} = (V,\overrightarrow{E})$ and a 
property $\mathcal{P}$, we want to test whether $\overrightarrow{G}$ satisfies the 
property or it is ``$\epsilon$-far" from satisfying the property. In the orientation model, the underlying graph $G = (V,E)$ is known in advance. Each edge in $E$ is 
oriented (that is directed in exactly one direction), and the orientation of the edges 
has to be queried. The graph is said to be ``$\epsilon$-far" from satisfying the 
property $\mathcal{P}$ if one has to reorient at least an $\epsilon$ fraction of the edges to 
make the graph satisfy the property. Here, the goal is to design a randomized algorithm that  
queries  the orientation of the edges, and does the following:
\begin{itemize} 
\item the algorithm ACCEPTS with probability at least $2/3$, if  $\overrightarrow{G}$ satisfies the property $\mathcal{P}$, 
\item the algorithm REJECTS with probability at least $2/3$, if $\overrightarrow{G}$ is ``$\epsilon$-far" from satisfying the property $\mathcal{P}$. 
\end{itemize}
The query complexity of the algorithm is defined by the number of edges  it queries.
The natural goal is to design a tester for $\mathcal{P}$ with  minimum query complexity. 
If the graph satisfies the property and the tester accepts with probability $1$, then 
the tester is called a \emph{$1$-sided error tester}. 
 The standard tester (as defined above) is called a \emph{$2$-sided 
error tester}.

The  orientation model for testing graph properties was introduced by Halevy \textit{et al.} \cite{halevy}.   
 Fischer \textit{et al.} studied the problem of testing whether an oriented graph $\overrightarrow{G}$ 
is Eulerian~\cite{euler}. They derived various upper 
and lower bounds for both  $1$-sided and $2$-sided error testers. 
Here, we use their algorithms for testing 
whether the uniform distribution is the stationary distribution of the Markov chain. 

\subsection{Testing whether the Uniform Distribution is a Stationary Distribution in the Orientation Model}\label{sec:testing}
Given a degree-$\Delta$ oriented graph $\overrightarrow{G}=(V, \overrightarrow{E}) $, we say that the graph 
has the property $\mathcal{P'}$ 
if for all $(u,v) \in \overrightarrow{E}$, we have  $d^+(u) = d^-(v).$
Since the underlying undirected graph is known in advance, we know its connected components. 
If the graph $\overrightarrow{G}$ is ``$\epsilon$-close" to satisfy the property $\mathcal{P'}$, then by 
{\em ``pigeonhole principle''} there is at  least one connected component of $\overrightarrow{G}$ that is also  
``$\epsilon$-close" to satisfy   the property 
$\mathcal{P'}$. Thus, if all connected components are ``$\epsilon$-far" from satisfying  the property 
$\mathcal{P'}$, then we can conclude that $\overrightarrow{G}$ is ``$\epsilon$-far" from satisfying the 
property $\mathcal{P'}$ (contrapositive of the previous statement). Thus, we  perform  testing on every 
connected-component of $\overrightarrow{G}$, and we reject
 if all connected components are ``$\epsilon$-far" from satisfying  the property 
$\mathcal{P'}$. Now, \textit{w.l.o.g.}, we will assume that the graph $\overrightarrow{G}$ is connected.

From Corollary~\ref{lemma:bipartite}, if $\overrightarrow{G}$ is non-bipartite then we have to 
test whether $\overrightarrow{G}$ is Eulerian. Since we can determine whether a graph is bipartite 
or not just by looking at the underlying undirected graph, if $\overrightarrow{G}$ is non-bipartite 
then we use the  Eulerianity testing algorithm from \cite{euler}.
 
 Now, let $\overrightarrow{G}$ be bipartite. Let the bipartition be $V_{L}$ and $V_{R}$. If $|V_{L}| \neq |V_{R}|$ 
then the graph surely does not satisfies the property $\mathcal{P'}$. From Corollary~\ref{lemma:bipartite}, 
if $|V_{L}| = |V_{R}|$ then the graph must have the property 
that the out-degree of all vertices in $V_{L}$ must be equal to the in-degree of all vertices in $V_{R}$ and vice versa. 
Let $v$ be a vertex in $V_{L}$ and  $d^-(v)=k_1$ and $d^+(v)= k_2$. Now, consider any bipartite directed graph 
$\overrightarrow{G^*} = (V, \overrightarrow{E^*})$ with bipartition $V_{L}$ and $V_{R}$ that satisfies the 
following conditions:
\vspace{-0.2 cm}
 \begin{itemize}
 \item The underlying undirected graphs of $\overrightarrow{G}$ and $\overrightarrow{G^*}$ are exactly the same, 
 \item $\forall u \in V_{L}$, $d^-_{\overrightarrow{G}^*}(u)=k_2$, $d^+_{\overrightarrow{G}^*}(u)=k_1$,  and
  $\forall v \in V_{R}$, $d^-_{\overrightarrow{G}^*}(v) = k_1$, $d^+_{\overrightarrow{G}^*}(v)=k_2$.
 \end{itemize}
\vspace{-0.2 cm}
Now consider the graph $\overrightarrow{G^{\cup}} = (V, \overrightarrow{E}\cup\overrightarrow{E^*})$ obtained 
by superimposing $\overrightarrow{G}$ and $\overrightarrow{G^*}$ - an edge $e\in \overrightarrow{E}^{\cup}$ if 
either $e\in  \overrightarrow{E}$, or $e\in \overrightarrow{E^*}$. 
Clearly, if $\overrightarrow{G}$ has the property $\mathcal{P'}$,
then $\overrightarrow{G^{\cup}}$ is Eulerian, 
and farness from having property $\mathcal{P'}$ is also true by the following lemma:
\begin{lemma}\label{prop}
If $\overrightarrow{G}$ is ``$\epsilon$-far" from having property $\mathcal{P'}$, then 
$\overrightarrow{G^{\cup}}$ is ``$\frac{\epsilon}{2}$-far"
from being Eulerian.
\end{lemma}
\begin{proof}
Let $\overrightarrow{H}$ be an Eulerian graph which is the closest to $\overrightarrow{G^{\cup}}$.
Since the underlying undirected graph for $\overrightarrow{G}$  and $\overrightarrow{G^*}$ are 
exactly the same, there is a one-to-one correspondence between the edges in $\overrightarrow{E}$ 
and $\overrightarrow{E^*}$. Now, look at edges of $\overrightarrow{G^{\cup}}$ that were flipped in order to obtain 
$\overrightarrow{H}$. Suppose a flipped edge (say $e$)  belonging to $\overrightarrow{E^*}$. Then, we can re-flip this edge $e$ and
flip the corresponding edge  in $\overrightarrow{E}$. Thus, we have effectively flipped the same number of edges.
By performing these operations on the flipped edges of $\overrightarrow{E^*}$, we have obtained a new graph
which has the same number of flipped edges as $\overrightarrow{G^{\cup}}$, and all the flipped edges in
$\overrightarrow{G^{\cup}}$ belongs to $\overrightarrow{E}$.

Thus, if the graph $\overrightarrow{G^{\cup}}$ is not ``$\frac{\epsilon}{2}$-far" from being Eulerian, then $\overrightarrow{G}$ is not
 ``$\epsilon$-far" from having property $\mathcal{P'}$, which is a contradiction.
\end{proof}

Now, all we have to test is whether the new graph $\overrightarrow{G^{\cup}}$ is Eulerian or ``$\frac{\epsilon}{2}$-far" from being Eulerian.
Note that every query to $\overrightarrow{G^{\cup}}$ can be simulated by a single query to 
$G$. Thus, we can now use the Eulerian testing algorithm from~\cite{euler}. 
The algorithm is summarized in  Algorithm $1$, and the various bounds on the query complexity that 
can be obtained is summarized in  Table $1$.

\begin{algorithm}[H]
\SetAlgoLined
\KwData{\emph{Degree}-$\Delta$ Oriented Graph $\overrightarrow{G}=(V,\overrightarrow{E})$}
\KwResult{Whether $\overrightarrow{G}$ has the property $\mathcal{P'}$
  or is ``$\epsilon$-far" from having it.}
\eIf{$G$ is non-bipartite}
{
  Test if $\overrightarrow{G}$ is Eulerian  (see \cite{euler}) and give the corresponding output.\\
}
{
  Let $V_L$ and $V_R$ be the bipartition for the graph $G$. \\
  
  Sample a vertex from $V_L$ and query all edges incident to it. Let $d^-(v)=k_1$ and $d^+(v)=k_2$.\\
  
  Construct any bipartite graph $\overrightarrow{G^*} = (V, \overrightarrow{E^*})$ with bipartition $(V_L, V_R)$ such that

  (a) For all $v \in V_{L}$, $d^-_{\overrightarrow{G}^*}(v)=k_2$ and $d^+_{\overrightarrow{G}^*}(v)=k_1$, and for all $v \in V_{R}$,  $d^-_{\overrightarrow{G}^*}(v) = k_1$ and $d^+_{\overrightarrow{G}^*}(v)=k_2$.

  (b) The underlying graph of $\overrightarrow{G^*}$ is exactly same as $G=(V,E)$. \\
  
  Superimpose $\overrightarrow{G^*}$ and the graph $G~ (\text{say~} \overrightarrow{G^{\cup}}=(V, E \cup\overrightarrow{E^*}))$.\\
  Test if $\overrightarrow{G^{\cup}}$ is Eulerian (see \cite{euler}) and give the corresponding output.
}
\caption{Algorithm for testing property $\mathcal{P'}$}
\end{algorithm}

\begin{table}
\begin{center}
\begin{tabular}{||l | c | c ||}
    \hline \hline
    &         &  \\
    & \textbf{1-sided test} & \textbf{2-sided test} \\ \hline
    &         &  \\
    \,\ Graphs with large $\Delta$\newline  & $\Delta$ + $O(\frac{m}{\epsilon^2\Delta})$  &  $\Delta$ + $\min\left\{\tilde{O}\left(\frac{m^3}{\epsilon^6\Delta^6}\right),\tilde{O}\left(\frac{m}{\epsilon^2\Delta^{\frac{3}{2}}}\right)\right\}$ \\

    &        &  \\ \hline
    &         &  \\
    \,\ Bounded-degree graphs $^*$ & $\Omega\left({m^{\frac{1}{4}}}\right)$  &  $\Omega \left(\sqrt{\frac{\log m}{\log\log m}}   \right)$\\

    &        &  \\ \hline
    &         &   \\
    \,\ $\alpha$-expander &  $O\left(\frac{\Delta\log \left(\frac{1}{\epsilon}\right)}{\alpha \epsilon}\right)$ &  $min\left\{\tilde{O}\left(\left(\frac{\log\left(\frac{1}{\epsilon}\right)}{\alpha\epsilon}\right)^3\right),\tilde{O}\left(\left(\frac{\sqrt{\Delta}\log\left(\frac{1}{\epsilon}\right)}{\alpha\epsilon}\right)\right)   \right\}$ \\
     &        &  \\
    \hline\hline
\end{tabular}
\end{center}
$^*$ Lower bound holds for $4$-regular graph.\\
$\Delta$ is the maximum degree of the underlying undirected graph and $m$ is the number of edges is the graph. \\ 

\caption{Bounds on the query complexity (in the orientation model) for testing uniformity of stationary distribution of 
        the Markov chain obtained by the random walk on a directed graph.}
\end{table}\label{table:bound}

%% file: conclusion.tex
\section{Conclusion and Open Problems }
We have shown that, for a given  \emph{degree}-$\Delta$ oriented graph $\overrightarrow{G} = (V, \overrightarrow{E})$, 
whether the uniform distribution on vertices of $\overrightarrow{G}$ is a stationary distribution 
of the Markov chain, depends on a local property of graph. If $\overrightarrow{G} $ satisfies 
this local property, then it has some particular kind of structure (see  
Corollary~\ref{lemma:bipartite}). Finally, as an application of this result, we showed that 
testing this local property in orientation model, can be reduced to testing Eulerianity 
 \cite{euler}.

  It is an interesting problem to test whether the stationary distribution of a Markov 
 chain is equal to some fixed distribution $D$. In this paper, we have considered  $D$ 
 to be uniform distribution, but  the same problem is also interesting for other 
 distributions. Moreover, finding a relationship between the distance of the stationary 
 distribution from the uniform distribution, and the number of edges that needs to be 
 reoriented, is an interesting open problem. Also, this result holds only for graphs where the 
 in-degree plus out-degree of all the vertices are  
same. A major open problem of this work is to come up with a similar statement for more general 
graphs.
